\documentclass[10pt,notitlepage,11pt,reqno]{amsart}

\usepackage{amssymb,amsmath,amsthm,mathrsfs,xspace,multicol,stmaryrd}
\usepackage[mathscr]{eucal}
\usepackage{amscd}
\usepackage[breaklinks=true]{hyperref}
\usepackage{url}

\newcommand{\comment}[1]{}
\newcommand{\blankbrac}{[ \cdot,\cdot ]}
\newcommand{\brac}[2]{\left \{ #1,#2 \right\}}

\newcommand{\cbrac}[2]{\left \llbracket #1,#2 \right \rrbracket}
\newcommand{\tcbrac}[2]{\left \llbracket #1,#2 \right \rrbracket_{\omega}}

\newcommand{\ham}{\mathrm{Ham}\left(M\right)}

\newcommand{\cinf}{C^{\infty}}

\newcommand{\lie}[2]{\mathcal{L}\,_{v_{#1}}\, #2}
\renewcommand{\L}{\mathcal{L}}

\newcommand{\ip}[1]{\iota_{v_{#1}}}

\newcommand{\Vect}{\mathrm{Vect}_{H}}

\newcommand{\D}{\mathcal{D}}

\newcommand{\innerprod}[2]{\langle #1,#2 \rangle}

\DeclareMathOperator{\ad}{\mathrm{ad}}

\newcommand{\tensor}{\otimes}

\newcommand{\maps}{\colon}

\newcommand{\half}{\frac{1}{2}}

\DeclareMathOperator{\id}{\mathrm{id}}
\DeclareMathOperator{\im}{\mathrm{im}}
\DeclareMathOperator{\cp}{\circlearrowleft}

\theoremstyle{plain}
\newtheorem{theorem}{Theorem}[section]
\newtheorem{prop}[theorem]{Proposition}
\newtheorem{lemma}[theorem]{Lemma}

\newtheorem{definition}[theorem]{Definition}

\theoremstyle{remark}

\newtheorem{example}{Example}

\title[Courant algebroids from categorified symplectic
geometry]{Courant algebroids from categorified symplectic geometry} \author{Christopher L.\ Rogers}
\email{\texttt{chris@math.ucr.edu}} \address{Department of
  Mathematics, University of California, Riverside, California 92521,
  USA} \date{\today} \thanks{This work was partially supported by a
  grant from The Foundational Questions Institute.}
\begin{document}

\begin{abstract}
  In categorified symplectic geometry, one studies the categorified
  algebraic and geometric structures that naturally arise on manifolds
  equipped with a closed nondegenerate $(n+1)$-form. The case
  relevant to classical string theory is when $n=2$ and is called
  `2-plectic geometry'. Just as the Poisson bracket makes the smooth
  functions on a symplectic manifold into a Lie algebra, there is a
  Lie 2-algebra of observables associated to any 2-plectic
  manifold. String theory, closed 3-forms and Lie 2-algebras also play
  important roles in the theory of Courant algebroids. Courant
  algebroids are vector bundles which generalize the structures found
  in tangent bundles and quadratic Lie algebras. It is known that a
  particular kind of Courant algebroid (called an exact Courant
  algebroid) naturally arises in string theory, and that such an
  algebroid is classified up to isomorphism by a closed 3-form on the
  base space, which then induces a Lie 2-algebra structure on the
  space of global sections. In this paper we begin to establish
  precise connections between 2-plectic manifolds and Courant
  algebroids. We prove that any manifold $M$ equipped
  with a 2-plectic form $\omega$ gives an exact Courant algebroid
  $E_{\omega}$ over $M$ with \v{S}evera class $[\omega]$, and we
  construct an embedding of the Lie 2-algebra of observables into the
  Lie 2-algebra of sections of $E_{\omega}$. We then show that this
  embedding identifies the observables as particular infinitesimal
  symmetries of $E_{\omega}$ which preserve the 2-plectic structure on
  $M$.
\end{abstract}

\maketitle

\section{Introduction}
\label{introduction}
The underlying geometric structures of interest in categorified
symplectic geometry are multisymplectic manifolds: manifolds equipped
with a closed, nondegenerate form of degree $\geq 2$
\cite{Cantrijn:1999}. This kind of geometry originated in the work of
DeDonder \cite{DeDonder} and Weyl \cite{Weyl} on the calculus of
variations, and more recently has been used as a formalism to
investigate classical field theories \cite{GIMM,Helein,Kijowski}.  In
this paper, we call a manifold `$n$-plectic' if it is equipped with a
closed nondegenerate $(n+1)$-form. Hence ordinary symplectic geometry
corresponds to the $n=1$ case, and the corresponding 1-dimensional
field theory is just the classical mechanics of point particles. In
general, examples of $n$-plectic manifolds include phase spaces
suitable for describing $n$-dimensional classical field theories. We
will be primarily concerned with the $n = 2$ case. This is the first
really new case of $n$-plectic geometry and the corresponding
$2$-dimensional field theories of interest include bosonic string
theory. Indeed, just as the phase space of the classical particle is a
manifold equipped with a closed, nondegenerate 2-form, the phase
space of the classical string is a \textit{finite-dimensional}
manifold equipped with a closed nondegenerate 3-form. This phase
space is often called the `multiphase space' of the string \cite{GIMM}
in order to distinguish it from the infinite-dimensional symplectic
manifolds that are used as phase spaces in string field theory
\cite{Bowick-Rajeev}.

In classical mechanics, the relevant mathematical structures are not
just geometric, but also algebraic. The symplectic form gives the
space of smooth functions the structure of Poisson
algebra. Analogously, in classical string theory, the 2-plectic form
induces a bilinear skew-symmetric bracket on a particular subspace of
differential 1-forms, which we call Hamiltonian. The Hamiltonian
1-forms and smooth functions form the underlying chain complex of an
algebraic structure known as a semistrict Lie 2-algebra.  A semistrict
Lie 2-algebra can be viewed as a categorified Lie algebra in which the
Jacobi identity is weakened and is required to hold only up to
isomorphism. Equivalently, it can be described as a 2-term
$L_{\infty}$-algebra, i.e.\ a generalization of a 2-term differential
graded Lie algebra in which the Jacobi identity is only satisfied up
to chain homotopy \cite{HDA6,LS}.  Just as the Poisson algebra of smooth
functions represents the observables of a system of particles, it has
been shown that the Lie 2-algebra of Hamiltonian 1-forms contains the
observables of the classical string \cite{BHR}. In general, an
$n$-plectic structure will give rise to a $L_{\infty}$-algebra on an
$n$-term chain complex of differential forms in which the $(n-1)$-forms
correspond to the observables of an $n$-dimensional classical field
theory \cite{CR}.
 
Many of the ingredients found in 2-plectic geometry are also found in
the theory of Courant algebroids, which was also developed by
generalizing structures found in symplectic geometry. Courant
algebroids were first used by Courant \cite{Courant} to study
generalizations of pre-symplectic and Poisson structures in the theory
of constrained mechanical systems. Roughly, a Courant algebroid is a
vector bundle that generalizes the structure of a tangent bundle
equipped with a symmetric nondegenerate bilinear form on the
fibers. In particular, the underlying vector bundle of a Courant
algebroid comes equipped with a skew-symmetric bracket on the space
of global sections. However, unlike the Lie bracket of vector fields,
the bracket need not satisfy the Jacobi identity.

In a letter to Weinstein, \v{S}evera \cite{Severa1} described how a
certain type of Courant algebroid, known as an exact Courant
algebroid, appears naturally when studying 2-dimensional variational
problems. In classical string theory, the string can be represented as
a map $\phi \maps \Sigma \to M$ from a 2-dimensional parameter space
$\Sigma$ into a manifold $M$ corresponding to space-time. The image
$\phi(\Sigma)$ is called the string world-sheet. The map $\phi$
extremizes the integral of a 2-form $\theta \in \Omega^{2}(M)$ over
its world-sheet. Hence the classical string is a solution to a
2-dimensional variational problem. The 2-form $\theta$ is called the
Lagrangian and depends on elements of the first jet bundle of the
trivial bundle $\Sigma \times M$. The Lagrangian is not unique. A
solution $\phi$ remains invariant if an exact 1-form or
`divergence' is added to $\theta$. It is, in fact, the 3-form
$d\theta$ that is relevant. In this context, \v{S}evera observed that
the 3-form $d\theta$ uniquely specifies (up to isomorphism) the
structure of an exact Courant algebroid over $M$. The general
correspondence between exact Courant algebroids and closed 3-forms on
the base space was further developed by \v{S}evera, and also by
Bressler and Chervov \cite{Bressler-Chervov}, to give a complete
classification.  An exact Courant algebroid over $M$ is determined up to
isomorphism by its \v{S}evera class: an element $[\omega]$ in the
third de Rham cohomology of $M$.

Just as in 2-plectic geometry, the underlying geometric structure of a
Courant algebroid has an algebraic manifestation. Roytenberg and
Weinstein \cite{Roytenberg-Weinstein} showed that the  bracket
on the space of global sections induces an $L_{\infty}$ structure. If
we are considering an exact Courant algebroid, then the global
sections can be identified with ordered pairs of vector fields and
1-forms on the base space. Roytenberg and Weinstein's results imply
that these sections, when combined with the smooth functions on the base
space, form a semistrict Lie 2-algebra \cite{Sheng-Zhu}. Moreover, the bracket of
the Lie 2-algebra is determined by a closed 3-form corresponding to a
representative of the \v{S}evera class
\cite{Severa-Weinstein}.

Thus there are striking similarities between 2-plectic manifolds and
exact Courant algebroids. Both originate from attempts to generalize
certain aspects of symplectic geometry. Both come equipped with a
closed 3-form that gives rise to a Lie 2-algebra structure on a chain
complex consisting of smooth functions and differential 1-forms. In
this paper, we prove that there is indeed a connection between the
two. We show that any manifold $M$ equipped with a 2-plectic form
$\omega$ gives an exact Courant algebroid $E_{\omega}$ with \v{S}evera
class $[\omega]$, and that there is an embedding of the Lie 2-algebra
of observables into the Lie 2-algebra corresponding to
$E_{\omega}$. Moreover, this embedding allows us to characterize the
Hamiltonian 1-forms as particular infinitesimal symmetries of
$E_{\omega}$ which preserve the 2-plectic structure on $M$.

\section{Courant algebroids}
Here we recall some basic facts and examples of Courant
algebroids and then we proceed to describe \v{S}evera's classification of exact Courant
algebroids. There are several equivalent definitions of a Courant
algebroid found in the literature. In this paper we use the
definition given by Roytenberg \cite{Roytenberg_thesis}.
\begin{definition} \label{courant_algebroid} A {\bf Courant algebroid}
  is a vector bundle $E \to M$ equipped with a nondegenerate symmetric
  bilinear form $\innerprod{\cdot}{\cdot}$ on the bundle, a
  skew-symmetric bracket $\cbrac{\cdot}{\cdot}$ on $\Gamma(E)$, and a
  bundle map (called the {\bf anchor}) $\rho \maps E \to TM$ such that
  for all $ e_{1},e_{2},e_{3}\in \Gamma (E)$ and for all $f,g \in
  \cinf(M)$ the following properties hold:
\begin{enumerate}
\item{$\cbrac{e_1}{\cbrac{e_2}{e_3}} - \cbrac{\cbrac{e_1}{e_2}}{e_3} 
-\cbrac{e_2}{\cbrac{e_1}{e_3}}=-\mathcal{D}T(e_{1},e_{2},e_{3}),$  
}

\item{$\rho([e_{1},e_{2}])=[\rho(e_{1}),\rho(e_{2})]$, \label{axiom2}
} 

\item{$[e_{1},fe_{2}]=f[e_{1},e_{2}]+\rho (e_{1})(f)e_{2}-\half
\langle e_{1},e_{2}\rangle {\mathcal{D}}f$,
}

\item{$\langle {\mathcal{D}}f,{\mathcal{D}}g\rangle =0$,}

\item{$\rho(e_{1})\left(\langle
e_{2},e_{3}\rangle\right) =\langle [e_{1},e_{2}]+\half {\mathcal{D}}\langle
e_{1},e_{2}\rangle ,e_{3}\rangle +\langle e_{2},[e_{1},e_{3}]+\half
{\mathcal{D}}\langle e_{1},e_{3}\rangle \rangle$,}
\end{enumerate}
where $[\cdot,\cdot]$ is the Lie bracket of vector fields,
$\D \maps \cinf(M) \to \Gamma(E)$ is the map defined by $\innerprod{\D f}{e}=\rho(e)f$, and 
\[
T(e_{1},e_{2},e_{3})= \frac{1}{6}
\left(\innerprod{\cbrac{e_1}{e_2}}{e_3} +
\innerprod{\cbrac{e_3}{e_1}}{e_2} + \innerprod{\cbrac{e_2}{e_3}}{e_1} \right). 
\]
\end{definition}
The bracket in Definition \ref{courant_algebroid} is skew-symmetric,
but the first property implies that it needs only to satisfy the
Jacobi identity ``up to $\D T$''. (The notation suggests we think of
this as a boundary.) The function $T$ is often referred to as the
\textbf{Jacobiator}. (When there is no risk of confusion, we shall
refer to the Courant algebroid with underlying vector bundle $E \to M$
as $E$.)

Note that the vector bundle $E$ may be identified with $E^{\ast}$ via the bilinear form
$\innerprod{\cdot}{\cdot}$ and therefore we have the dual map
\[ \rho^{\ast} \maps T^{\ast}M \to E.\] 
Hence the map $\D$ is simply the pullback of the de Rham differential by $\rho^{\ast}$. 

There is an alternate definition given by \v{S}evera \cite{Severa1}
for Courant algebroids which uses a bilinear operation on sections
that satisfies a Jacobi identity but is not skew-symmetric.
\begin{definition}\label{alt_courant_algebroid}
  A {\bf Courant algebroid} is a vector bundle $ E\rightarrow M $
  together with a nondegenerate symmetric bilinear form $ \innerprod
  {\cdot }{\cdot } $ on the bundle, a bilinear operation $ \circ $ on
  $ \Gamma (E) $, and a bundle map $ \rho \maps E\rightarrow TM $ such
  that for all $ e_{1},e_{2},e_{3}\in \Gamma (E)$ and for all $f \in
  \cinf(M)$ the following properties hold:
\begin{enumerate}
\item{$ e_{1}\circ (e_{2}\circ e_{3})=(e_{1}\circ e_{2})\circ e_{3}+e_{2}\circ (e_{1}\circ e_{3})$,}
\item {$\rho (e_{1}\circ e_{2})=[\rho (e_{1}),\rho (e_{2})]$,} 
\item {$ e_{1}\circ fe_{2}=f(e_{1}\circ e_{2})+\rho (e_{1})(f)e_{2}$,}
\item {$ e_{1}\circ e_{1}=\half \D \innerprod {e_{1}}{e_{1}}$,}
\item {$ \rho (e_{1})\left(\innerprod {e_{2}}{e_{3}}\right)=\innerprod {e_{1}\circ
      e_{2}}{e_{3}}+\innerprod {e_{2}}{e_{1}\circ e_{3}}$,}
\end{enumerate}
where $[\cdot,\cdot]$ is the Lie bracket of vector fields, and
$\D \maps \cinf(M) \to \Gamma(E)$ is the map defined by $\innerprod{\D f}{e}=\rho(e)f$.
\end{definition}
The ``bracket'' $\circ$ is related to the bracket given in
Definition \ref{courant_algebroid} by: 
\begin{equation}
 x \circ y = \cbrac{x}{y} + \half \D
 \innerprod{x}{y}. \label{dorfman}
\end{equation}
Roytenberg \cite{Roytenberg_thesis} showed that if $E$ is a Courant
algebroid in the sense of Definition \ref{courant_algebroid} with bracket
$\cbrac{\cdot}{\cdot}$, bilinear form $\innerprod{\cdot}{\cdot}$ and
anchor $\rho$, then $E$ is a Courant algebroid in the sense of Definition
\ref{alt_courant_algebroid} with the same anchor and bilinear form but
with bracket $\circ$ given by Eq. \ref{dorfman}. Unless otherwise
stated, all Courant algebroids mentioned in this paper are Courant
algebroids in the sense of Definition \ref{courant_algebroid}. We
introduced Definition \ref{alt_courant_algebroid} mainly to connect our
results here with previous results in the literature. 

\begin{example} \label{standard}
An important example of a Courant algebroid is the \textbf{standard
  Courant algebroid} $E_{0}=TM \oplus T^{\ast}M$ over any manifold
  $M$ with bracket
\begin{equation} 
\cbrac{(v_{1},\alpha_{1})}{(v_{2},\alpha_{2})}_{0}=
\left([v_{1},v_{2}], \L_{v_{1}}\alpha_{2} - \L_{v_{2}}\alpha_{1} -
\half d\left(\ip{1}\alpha_{2} - \ip{2}\alpha_{1} \right) \right), \label{standard_bracket}
 \end{equation}
and bilinear form
\begin{equation}
\innerprod{(v_{1},\alpha_{1})}{(v_{2},\alpha_{2})}=
\ip{1}\alpha_{2} + \ip{2}\alpha_{1}.\label{standard_innerprod}
\end{equation}
In this case the anchor $\rho \maps E_{0} \to TM$ is the
projection map, and for a function $f \in \cinf(M)$, $\D f = (0,df)$.
\end{example}

The standard Courant algebroid is the prototypical example of
an \textbf{exact Courant algebroid} \cite{Bressler-Chervov}.
\begin{definition} \label{exact} A Courant algebroid $E \to M$ with
  anchor map $\rho \maps E \to TM$ is {\bf exact} iff
\[
0 \to T^{\ast}M \stackrel{\rho^{\ast}}{\to} E \stackrel{\rho}{\to} TM
\to 0\]
is an exact sequence of vector bundles.
\end{definition}

\subsection{The \v{S}evera class of an exact Courant algebroid} \label{class}
\v{S}evera's classification originates in the idea that choosing a
splitting of the above short exact sequence corresponds to defining a
kind of connection.
\begin{definition}\label{connection}
A {\bf connection} on an exact Courant algebroid $E$ over a manifold $M$ is a map of
vector bundles $A \maps TM \to E $ such that
\begin{enumerate}
\item{ $\rho \circ A = \id_{TM}$,}
\item{$\innerprod{A(v_{1})}{A(v_{2})}=0$ for all $v_{1},v_{2} \in TM$,}
\end{enumerate}
where $\rho \maps E \to TM$ and $\innerprod{\cdot}{\cdot}$ are the
anchor and bilinear form, respectively.
\end{definition}
If $A$ is a connection and $\theta \in
\Omega^2(M)$ is a 2-form then one can construct a new connection:
\begin{equation}
\left(A+\theta\right)(v)=A(v)+ \rho^{\ast}\theta(v,\cdot). \label{2-form_action}
\end{equation}
$(A+\theta)$ satisfies the first condition of Definition \ref{connection} since
$\ker \rho=\im \rho^{\ast}$. The second condition follows from the
fact that we have by definition of $\rho^{\ast}$:
\begin{equation}
\innerprod{\rho^{\ast}(\alpha)}{e}=\alpha(\rho(e)) \label{innerprod}
\end{equation}
for all $e \in \Gamma(E)$ and $\alpha \in \Omega^{1}(M)$.
Furthermore, one can show that any two connections on an exact Courant
algebroid must differ (as in Eq.\ \ref{2-form_action}) by a 2-form on
$M$. Hence the space of connections on an exact Courant algebroid
is an affine space modeled on the vector space of 2-forms
$\Omega^{2}(M)$ \cite{Bressler-Chervov}.

The failure of a connection to preserve the bracket gives a suitable
notion of curvature:
\begin{definition}
  If $E$ is an exact Courant algebroid over $M$ with bracket
  $\cbrac{\cdot}{\cdot}$ and $A \maps TM \to E$ is a connection then
  the {\bf curvature} is a map $F \maps TM \times TM \to E$ defined
  by
\[ F(v_{1},v_{2}) = \cbrac{A(v_{1})}{A(v_{2})} - A\left( \left[
    v_{1},v_{2} \right] \right).
\]
\end{definition}
If $F$ is the curvature of a connection $A$ then given $v_{1},v_{2}
\in TM$, it follows from exactness and axiom \ref{axiom2} in Definition
\ref{courant_algebroid} that there exists a 1-form
$\alpha_{v_{1},v_{2}} \in \Omega^1(M)$ such that
$F(v_{1},v_{2})=\rho^{\ast}(\alpha_{v_{1},v_{2}})$. Since $A$ is a
connection, its image is isotropic in $E$. Therefore for any $v_{3}
\in TM$ we have:
\[
\innerprod{F(v_{1},v_{2})}{A(v_{3})} = 
\innerprod{\cbrac{A\left(v_{1}\right)}{A\left(v_{2}\right)}}{A(v_{3})}.
\] 
The above formula allows one to associate the curvature $F$ to a 3-form on $M$:  
\begin{prop}
  Let $E$ be an exact Courant algebroid over a manifold $M$ with
  bracket $\cbrac{\cdot}{\cdot}$ and bilinear form
  $\innerprod{\cdot}{\cdot}$. Let $A \maps TM \to E$ be a connection on $E$. Then
  given vector fields $v_{1}, v_{2}, v_{3}$ on $M$:
\begin{enumerate}
\item{ 
The function
\[
\omega(v_{1},v_{2},v_{3})= \innerprod{\cbrac{A\left(v_{1}\right)}{A\left(v_{2}\right)}}{A(v_{3})}
\] 
defines a closed 3-form on $M$.}
\item{If $\theta \in \Omega^{2}(M)$ is a 2-form and
    $\tilde{A}=A+\theta$ then
\begin{align*}
\tilde{\omega}(v_{1},v_{2},v_{3}) &=
\innerprod{\llbracket
  \tilde{A}\left(v_{1}\right),\tilde{A}\left(v_{2}\right) \rrbracket}{\tilde{A}(v_{3})}\\
&=\omega(v_{1},v_{2},v_{3}) + d\theta(v_{1},v_{2},v_{3}).
\end{align*}
}
\end{enumerate}
\end{prop}
\begin{proof}
  The statements in the proposition are proven in Lemmas 4.2.6, 4.2.7,
  and 4.3.4 in the paper by Bressler and Chervov
  \cite{Bressler-Chervov}.  In their work they define a Courant
  algebroid using Definition \ref{alt_courant_algebroid}, and therefore
  their bracket satisfies the Jacobi identity, but is not
  skew-symmetric.  In our notation, their definition of the curvature
  3-form is:
\[
\omega^{\prime}(v_{1},v_{2},v_{3})= \innerprod{A\left(v_{1}\right)
  \circ A\left(v_{2}\right)}{A(v_{3})}.
\] 
In particular they show that $\circ$
satisfying the Jacobi identity implies $\omega^{\prime}$ is
closed. The Jacobiator corresponding to the Courant bracket is
non-trivial in general. However the isotropicity of the connection and
Eq.\ \ref{dorfman} imply
\[
A(v_1) \circ A(v_2)  = \cbrac{A(v_1)}{A(v_2)} \quad \forall
   v_1, v_2 \in TM.
\]
Hence $\omega^{\prime}=\omega$, so all the needed results in
\cite{Bressler-Chervov} apply here.
\end{proof}

Thus the above  proposition implies that the curvature 3-form of an exact
Courant algebroid over $M$ gives a well-defined cohomology class in
$H^{3}_{\mathrm{DR}}(M)$, independent of the choice of
connection.   

\subsection{Twisting the Courant bracket} \label{twist}
The previous section describes how to go from exact Courant algebroids
to closed 3-forms. Now we describe the reverse
process. In Example \ref{standard} we showed that one can define the
standard Courant algebroid $E_{0}$ over any manifold $M$. The total
space is the direct sum $TM \oplus T^{\ast}M$, the bracket and
bilinear form are given in Eqs.\ \ref{standard_bracket} and
\ref{standard_innerprod}, and the anchor is simply the projection.
The inclusion $A(v)=(v,0)$ of the tangent bundle into
the direct sum is obviously a connection on $E_{0}$ and it is easy to
see that the standard Courant algebroid has zero curvature.

\v{S}evera and Weinstein \cite{Severa1,Severa-Weinstein} observed that
the bracket on $E_{0}$ could be twisted by a closed 3-form
$\omega \in \Omega^3(M)$ on the base: 
\[
\cbrac{(v_{1},\alpha_{1})}{(v_{2},\alpha_{2})}_{\omega}=
\cbrac{(v_{1},\alpha_{1})}{(v_{2},\alpha_{2})}_{0} +
\omega(v_{1},v_{2},\cdot). \label{twisted_bracket} 
\]
This gives a new Courant algebroid $E_{\omega}$ with the same anchor
and bilinear form. Using Eqs.\ \ref{standard_bracket} and
\ref{standard_innerprod} we can compute the curvature 3-form of this
new Courant algebroid:
\begin{align*}
\innerprod{\cbrac{A\left(v_{1}\right)}{A\left(v_{2}\right)}}{A(v_{3})}
&=\innerprod{\cbrac{(v_{1},0)}{(v_{2},0)}}{(v_{3},0)}\\
&=\innerprod{\left([v_{1},v_2],\omega(v_{1},v_{2},\cdot)\right)}{(v_{3},0)}\\
&=\omega(v_{1},v_{2},v_{3}),
\end{align*}
and we see that $E_{\omega}$ is an exact Courant algebroid over $M$ with
\v{S}evera class $[\omega]$.

\section{$2$-plectic geometry}
We now give a brief overview of $2$-plectic geometry. More details
including motivation for several of the definitions presented here can be found
in our previous work with Baez and Hoffnung \cite{BHR,BR}.
\begin{definition}
\label{2-plectic_def}
A $3$-form $\omega$ on a $C^\infty$ manifold $M$ is 
{\bf 2-plectic}, or more specifically
a {\bf 2-plectic structure}, if it is both closed:
\[
    d\omega=0,
\]
and nondegenerate:
\[
    \forall v \in T_{x}M,\ \iota_{v} \omega =0 \Rightarrow v =0
\]
If $\omega$ is a $2$-plectic form on $M$ we call the pair $(M,\omega)$ 
a {\bf 2-plectic manifold}.
\end{definition}

The $2$-plectic structure induces an injective map from the
space of vector fields on $M$ to the space of 2-forms on $M$. This leads
us to the following definition:

\begin{definition} \label{hamiltonian}
Let $(M,\omega)$ be a $2$-plectic manifold.  A 1-form $\alpha$ on $M$ 
is {\bf Hamiltonian} if there exists a vector field $v_\alpha$ on $M$ such that
\[
d\alpha= -\ip{\alpha} \omega.
\]
We say $v_\alpha$ is the {\bf Hamiltonian vector field} corresponding to $\alpha$. 
The set of Hamiltonian 1-forms and the set of Hamiltonian vector
fields on a $2$-plectic manifold are both vector spaces and are denoted
as $\ham$ and $\Vect\left(M \right)$, respectively.
\end{definition}

The Hamiltonian vector field $v_\alpha$ is unique if
it exists, but note there may
be 1-forms $\alpha$ having no Hamiltonian vector field.  
Furthermore, two distinct Hamiltonian 1-forms may differ by a closed
1-form and therefore share the same Hamiltonian vector field.

We can generalize the Poisson bracket of functions in symplectic geometry by  
defining a bracket of Hamiltonian 1-forms. 
\begin{definition}
\label{semi-bracket.defn}
Given $\alpha,\beta\in \ham$, the {\bf bracket} $\brac{\alpha}{\beta}$
is the 
\break
1-form given by 
\[  \brac{\alpha}{\beta} = \ip{\beta}\ip{\alpha}\omega .\]
\end{definition}

\begin{prop}\label{semi-bracket} Let $\alpha,\beta,\gamma \in \ham$ and let
$v_\alpha,v_\beta,v_\gamma$ be the respective Hamiltonian
vector fields.  The bracket $\brac{\cdot}{\cdot}$ has the following properties:
  \begin{enumerate}
\item The bracket of Hamiltonian forms is Hamiltonian:
  \begin{eqnarray}
    d\brac{\alpha}{\beta} = -\iota_{[v_\alpha,v_\beta]} \omega,
  \end{eqnarray}
so in particular we have 
\[     v_{\brac{\alpha}{\beta}} = [v_\alpha,v_\beta]  .\]
\item The bracket is skew-symmetric: 
  \begin{eqnarray}
    \brac{\alpha}{\beta} = -\brac{\beta}{\alpha}
  \end{eqnarray}

\item The bracket satisfies the Jacobi identity up to an exact 1-form:
\begin{eqnarray}
    \brac{\alpha}{\brac{\beta}{\gamma}} -
    \brac{\brac{\alpha}{\beta}}{\gamma} 
    -\brac{\beta}{\brac{\alpha}{\gamma}} =dJ_{\alpha,\beta,\gamma} 
  \end{eqnarray}
with $J_{\alpha,\beta,\gamma}=\ip{\alpha}\ip{\beta}\ip{\gamma}\omega$.
\end{enumerate}
\end{prop}
\begin{proof}
See Proposition 3.7 in \cite{BHR}.
\end{proof}

\section{Lie $2$-algebras}
Both the Courant bracket and the bracket on Hamiltonian 1-forms are,
roughly, Lie brackets which satisfy the Jacobi identity up to an exact
1-form. This leads us to the notion of a Lie $2$-algebra: a category
equipped with structures analogous to those of a Lie algebra, for
which the usual laws involving skew-symmetry and the Jacobi identity
hold up to isomorphism \cite{HDA6,Roytenberg_L2A}. A Lie 2-algebra in
which the isomorphisms are actual equalities is called a strict Lie
2-algebra. A Lie 2-algebra in which the laws governing skew-symmetry
are equalities but the Jacobi identity holds only up to isomorphism is
called a semistrict Lie 2-algebra.

Here we define a semistrict Lie $2$-algebra to be a $2$-term chain
complex of vector spaces equipped with structures analogous to those of a
Lie algebra, for which the usual laws hold up to chain homotopy. In
this guise, a semistrict Lie $2$-algebra is nothing more than a
$2$-term $L_{\infty}$-algebra. For more details, we refer the reader
to the work of Lada and Stasheff \cite{LS}, and the work of Baez and
Crans \cite{HDA6}.
\begin{definition}
A {\bf semistrict Lie 2-algebra} is a $2$-term chain complex of vector spaces
$L = (L_1\stackrel{d}\rightarrow L_0)$ equipped with:
\begin{itemize}
\item a chain map $\blankbrac\maps L \tensor L\to L$ called the {\bf
bracket};
\item an antisymmetric chain homotopy $J \maps L \tensor L \tensor L
  \to L$ 
from the chain map
\[     \begin{array}{ccl}  
     L \tensor L \tensor L & \to & L   \\
     x \tensor y \tensor z & \longmapsto & [x,[y,z]],  
  \end{array}
\]
to the chain map
\[     \begin{array}{ccl}  
     L \tensor L \tensor L& \to & L   \\
     x \tensor y \tensor z & \longmapsto & [[x,y],z] + [y,[x,z]]  
  \end{array}
\]
called the {\bf Jacobiator},
\end{itemize}
such that the following equation holds:
\begin{equation}
\begin{array}{c}
  [x,J(y,z,w)] + J(x,[y,z],w) +
  J(x,z,[y,w]) + [J(x,y,z),w] \\ + [z,J(x,y,w)] 
  = J(x,y,[z,w]) + J([x,y],z,w) \\ + [y,J(x,z,w)] + J(y,[x,z],w) + J(y,z,[x,w]).
\end{array}
\end{equation}
\end{definition}

We will also need a suitable notion of morphism: 
\begin{definition}
\label{homo}
Given semistrict Lie $2$-algebras $L$ and $L'$ with bracket and 
Jacobiator $\blankbrac$, $J$ and $\blankbrac^\prime$,
$J^\prime$ respectively, a {\bf homomorphism} from $L$ to $L'$
consists of:
\begin{itemize}
\item{a chain map $\phi=\left(\phi_{0},\phi_{1}\right) \maps L \to L'$, and}
\item{a chain homotopy $\phi_{2} \maps L \tensor L \to L$ from the chain
  map
\[     \begin{array}{ccl}  
     L \tensor L & \to & L   \\
     x \tensor y & \longmapsto & \left [ \phi(x),\phi(y) \right]^{\prime},
  \end{array}
\]
to the chain map
\[     \begin{array}{ccl}  
     L \tensor L & \to & L   \\
     x \tensor y & \longmapsto & \phi \left( [x,y] \right)
  \end{array}
\]
}
\end{itemize}
such that the following equation holds:
\begin{equation} \label{coherence}
\begin{array}{l}
 J^{\prime}(\phi_0(x),\phi_0(y), \phi_0(z))-\phi_1(J(x,y,z)) = \\
\phi_2(x,[y,z]) -\phi_2([x,y],z) - \phi_2(y,[x,z]) - [\phi_2(x,y),\phi_0(z)]^{\prime}\\
+ [\phi_0(x), \phi_2(y,z)]^{\prime}- [\phi_0(y),\phi_2(x,z)]^{\prime}.
\end{array}
\end{equation}
\end{definition}
This definition is equivalent to the definition of a morphism between
$2$-term $L_{\infty}$-algebras. (The same definition is given
in \cite{HDA6}, but it contains a typographical error.)

\subsection{The Lie $2$-algebra from a $2$-plectic manifold}
Given a $2$-plectic manifold $(M,\omega)$, we can 
construct a semistrict Lie $2$-algebra. The underlying $2$-term chain complex is 
namely:
\[
L \quad = \quad 
\cinf(M)  \stackrel{d}{\rightarrow} \ham  
\]
where $d$ is the usual exterior derivative of functions.  
This chain complex is well-defined, since 
any exact form is Hamiltonian, with $0$ as its Hamiltonian vector
field. We can construct a chain map 
\[      \brac{\cdot}{\cdot} \maps L \otimes L \to L ,\]
by extending the bracket on $\ham$ trivially to $L$.
In other words, in degree $0$, the chain map is given as in 
Definition \ \ref{semi-bracket.defn}: 
\[  \brac{\alpha}{\beta} = \ip{\beta}\ip{\alpha}\omega, \]
and in degrees $1$ and $2$, we set it equal to zero:
\[  \brac{\alpha}{f} = 0, \qquad \brac{f}{\alpha} = 0, \qquad
    \brac{f}{g} = 0.    \]
The precise construction of this Lie $2$-algebra is 
given in the following theorem:
\begin{theorem}
\label{semistrict}
If $(M,\omega)$ is a $2$-plectic manifold, there is a 
semistrict Lie $2$-algebra $L(M,\omega)$ where:
\begin{itemize}
\item the space of 0-chains is $\ham$,
\item the space of 1-chains is $\cinf$,
\item the differential is the exterior derivative $d \maps \cinf \to \ham$,
\item the bracket is $\brac{\cdot}{\cdot}$,
\item the Jacobiator is the linear map $J_{L}\maps \ham \tensor \ham \tensor 
\ham \to \cinf$ defined by $J_{L}(\alpha,\beta,\gamma) = \ip{\alpha}\ip{\beta}\ip{\gamma}\omega$.
\end{itemize}
\end{theorem}
\begin{proof}
See Theorem 4.4 in \cite{BHR}.
\end{proof}
\subsection{The Lie $2$-algebra from a Courant algebroid}
Given any Courant algebroid $E \rightarrow M$ with bilinear form
$\innerprod{\cdot}{\cdot}$, bracket $\cbrac{\cdot}{\cdot}$, and anchor
$\rho \maps E \to TM$, we can construct a
$2$-term chain complex 
\[
C \quad = \quad \cinf(M) \stackrel{\D}{\rightarrow} \Gamma(E),
\]
with differential $\D=\rho^{\ast}d$. The bracket
$\cbrac{\cdot}{\cdot}$ on global sections can be extended to a chain
map $\cbrac{\cdot}{\cdot} \maps C \tensor C \to C$. If $e_1,e_2$ are
degree 0 chains then $\cbrac{e_{1}}{e_{2}}$ is the original bracket.
If $e$ is a degree 0 chain and $f,g$ are degree 1 chains, then we
define:
\begin{align*}
\cbrac{e}{f} &= -\cbrac{f}{e} = \half \innerprod{e}{\D f}  \\
\cbrac{f}{g}&=0.
\end{align*}   
This extended bracket gives a semistrict Lie $2$-algebra on the complex $C$:
\begin{theorem}\label{courant_L2A}
If $E$ is a Courant algebroid, there is a semistrict Lie $2$-algebra  
$C(E)$ where:
\begin{itemize}
\item the space of 0-chains is $\Gamma(E)$,
\item the space of 1-chains is $\cinf(M)$,
\item the differential the map $\D \maps \cinf(M) \to \Gamma(M)$,
\item the bracket is $\cbrac{\cdot}{\cdot}$,
\item the Jacobiator is the linear map $J_{C}\maps \Gamma(M) \tensor \Gamma(M) \tensor 
\Gamma(M) \to \cinf(M)$ defined by
\begin{align*} 
J_{C}(e_{1},e_{2},e_{3}) &= -T(e_{1},e_{2},e_{3})\\
&=-\frac{1}{6}
\left(\innerprod{\cbrac{e_1}{e_2}}{e_3} +
\innerprod{\cbrac{e_3}{e_1}}{e_2} + \innerprod{\cbrac{e_2}{e_3}}{e_1} \right). 
\end{align*}
\end{itemize}
\end{theorem}
\begin{proof}
The proof that a Courant algebroid in the sense of Definition
\ref{courant_algebroid} gives rise to a semistrict Lie 2-algebra follows from the work
done by Roytenberg on graded symplectic supermanifolds
\cite{Roytenberg_graded} and Lie $2$-algebras
\cite{Roytenberg_L2A}. In particular we refer the reader to Example
5.4 of \cite{Roytenberg_L2A} and Section 4 of
\cite{Roytenberg_graded}.

On the other hand, the original construction of Roytenberg
and Weinstein \cite{Roytenberg-Weinstein} gives a $L_{\infty}$-algebra 
on the complex:
\[
0 \rightarrow \ker \D \stackrel{\iota}{\rightarrow} \cinf(M)
\stackrel{\D}{\rightarrow}\Gamma(E),
\]
with trivial structure maps $l_{n}$ for $n \geq 3$. Moreover, the map
$l_{2}$ (corresponding to the bracket $\cbrac{\cdot}{\cdot}$ given above) is trivial in
degree $>1$ and the map $l_{3}$ (corresponding to the Jacobiator $J_{C}$) is
trivial in degree $>0$. Hence we can restrict this
$L_{\infty}$-algebra to our complex $C$ and use the results
in \cite{HDA6} that relate $L_{\infty}$-algebras with semistrict Lie
$2$-algebras.
\end{proof}  

\section{The Courant algebroid associated to a $2$-plectic manifold}
Now we have the necessary machinery in place to describe how
Courant algebroids connect with $2$-plectic geometry.  First, recall the
discussion in Section \ref{twist} on twisting the bracket of the
standard Courant algebroid $E_{0}$ by a closed
3-form. From Definition \ref{2-plectic_def}, we immediately have the following:
\begin{prop} \label{2-plectic_courant}
Let $(M,\omega)$ be a $2$-plectic manifold. There exists an exact Courant
algebroid $E_{\omega}$ with \v{S}evera class $[\omega]$ over $M$ with underlying vector bundle $TM
\oplus T^{\ast}M \rightarrow M$, anchor $\rho(v,\alpha)=v$,
and bracket and bilinear form given by:
\[
\cbrac{(v_{1},\alpha_{1})}{(v_{2},\alpha_{2})}_{\omega}=
\left([v_{1},v_{2}], \L_{v_{1}}\alpha_{2} - \L_{v_{2}}\alpha_{1} -
\half d\left(\ip{1}\alpha_{2} - \ip{2}\alpha_{1} \right) +
\ip{2}\ip{1} \omega \right) , 
 \]
\[ \innerprod{(v_{1},\alpha_{1})}{(v_{2},\alpha_{2})}=
\ip{1}\alpha_{2} + \ip{2}\alpha_{1}.
\]
\end{prop}

More importantly, the Courant algebroid constructed in Proposition \ref{2-plectic_courant}
not only encodes the $2$-plectic structure $\omega$, but also the
corresponding Lie $2$-algebra  constructed in Theorem
\ref{semistrict}:

\begin{theorem} \label{main_thm} Let $(M,\omega)$ be a $2$-plectic
  manifold and let $E_{\omega}$ be its corresponding Courant
  algebroid. Let $L(M,\omega)$ and $C(E_{\omega})$ be the semistrict
  Lie 2-algebras corresponding to $(M,\omega)$ and $E_{\omega}$,
  respectively.  Then there exists a homomorphism embedding
  $L(M,\omega)$ into $C(E_{\omega})$.
\end{theorem} 

Before we prove the theorem, we introduce some lemmas to ease the calculations. In the
notation that follows, if $\alpha, \beta$ are Hamiltonian 1-forms with
corresponding vector fields $v_{\alpha},v_{\beta}$, then
\begin{equation}
B(\alpha,\beta)=\half(\ip{\alpha}\beta - \ip{\beta}\alpha).
\end{equation}
Also by the symbol $\cp$ we mean cyclic permutations of the symbols $\alpha,\beta,\gamma$.
\begin{lemma}\label{calc1}
If $\alpha, \beta \in \ham$ with corresponding Hamiltonian vector fields
$v_{\alpha},v_{\beta}$, then
$\L_{v_{\alpha}}\beta=\brac{\alpha}{\beta} + d \ip{\alpha}\beta$.
\end{lemma}
\begin{proof} Since $\L_{v} = \iota_v d + d \iota_v$,
\[  \L_{v_{\alpha}}{\beta} = 
\ip{\alpha} d \beta + d \ip{\alpha} \beta =
-\ip{\alpha}\ip{\beta} \omega + d \ip{\alpha} \beta =
\brac{\alpha}{\beta} + d \ip{\alpha} \beta .\]
\end{proof}

\begin{lemma} \label{calc2}
If $\alpha,\beta,\gamma \in \ham$ with corresponding Hamiltonian vector fields
$v_{\alpha},v_{\beta},v_{\gamma}$, then
\[
\iota_{[v_{\alpha},v_{\beta}]} \gamma + \cp = 
-3\ip{\alpha}\ip{\beta}\ip{\gamma}\omega + 2 \left
  (\ip{\alpha}dB(\beta,\gamma) + \ip{\gamma}dB(\alpha,\beta) +
  \ip{\beta}dB(\gamma,\alpha) \right).
\]
\end{lemma}
\begin{proof}
The identity $\iota_{[v_{1},v_{2}]}=\L_{v_{1}}\ip{2}
-\ip{2}\L_{v_{1}}$ and Lemma \ref{calc1} imply:
\begin{align*}
\iota_{[v_{\alpha},v_{\beta}]} \gamma &=
\lie{\alpha}{\ip{\beta}}\gamma - {\ip{\beta}}\lie{\alpha}\gamma\\
&=\lie{\alpha}{\ip{\beta}}\gamma -\ip{\beta}
\left(\brac{\alpha}{\gamma} + d \ip{\alpha}\gamma\right)\\
&=\ip{\alpha}d\ip{\beta}\gamma -\ip{\beta}\ip{\gamma}\ip{\alpha}\omega 
  - \ip{\beta}d\ip{\alpha}\gamma,
\end{align*}
where the last equality follows from the definition of the bracket.

Hence it follows that: 
\begin{align*}
\iota_{[v_{\gamma},v_{\alpha}]} \beta  
&=\ip{\gamma}d\ip{\alpha}\beta -\ip{\alpha}\ip{\beta}\ip{\gamma}\omega 
  - \ip{\alpha}d\ip{\gamma}\beta,\\
\iota_{[v_{\beta },v_{\gamma}]}\alpha  
&=\ip{\beta}d\ip{\gamma}\alpha -\ip{\gamma}\ip{\alpha}\ip{\beta}\omega 
  - \ip{\gamma}d\ip{\beta}\alpha.
\end{align*}

Finally, note $2 \ip{\alpha}dB(\beta,\gamma)=\ip{\alpha}d\ip{\beta}\gamma -
\ip{\alpha}d\ip{\gamma}\beta$. 
\end{proof}

\begin{lemma}\label{calc3}
If $\alpha, \beta \in \ham$ with corresponding Hamiltonian vector fields
$v_{\alpha},v_{\beta}$, then
\[
\lie{\beta}{\alpha}-\lie{\alpha}{\beta}=-2\left(\brac{\alpha}{\beta}+dB(\alpha,\beta)
  \right).
\]
\end{lemma}
\begin{proof}
Follows immediately from Lemma \ref{calc1} and the definition of $B(\alpha,\beta)$.
\end{proof}

\begin{proof}[Proof of Theorem \ref{main_thm}]
We will construct a homomorphism from $L(M,\omega)$ to $C(E_{\omega})$.
Let $L$ be the underlying chain complex of $L(M,\omega)$
consisting of Hamiltonian 1-forms in degree 0 and smooth functions in
degree 1. Let $C$ be the underlying chain of $C(E_{\omega})$
consisting of global sections of $E_{\omega}$ in degree 0 and smooth
functions in degree 1. The bracket $\tcbrac{\cdot}{\cdot}$ denotes the
extension of the bracket on $\Gamma(E_{\omega})$ to the complex $C$ in the
sense of Theorem \ref{courant_L2A}.
Let $\phi_{0} \maps L_{0} \to C_{0}$ be given by
\[
\phi_{0}(\alpha)=\left(v_\alpha, -\alpha \right),
\]
where $v_{\alpha}$ is the Hamiltonian vector field corresponding to
$\alpha$.
Let $\phi_{1}\maps L_{1} \to C_{1}$ be given by
\[
\phi_{1}(f)=-f.
\]
Finally let $\phi_{2} \maps L_{0}\tensor L_{0}\to C_{1}$ be given
by
\[
\phi_{2}(\alpha,\beta)=-B(\alpha,\beta)=-\half(\ip{\alpha}\beta-
\ip{\beta}\alpha).
\]

Now we show $\phi_{2}$ is a well-defined chain homotopy in the sense of
Definition \ref{homo}. For degree 0 we have:
\begin{align*}
\tcbrac{\phi_{0}(\alpha)}{\phi_{0}(\beta)}&=
\left([v_{\alpha},v_{\beta}], \L_{v_{\alpha}}(-\beta) - \L_{v_{\beta}}(-\alpha) +
\half d\left(\ip{\alpha}\beta - \ip{\beta}\alpha \right) +
\ip{\beta}\ip{\alpha} \omega \right) \\
&=\left([v_{\alpha},v_{\beta}], -\brac{\alpha}{\beta} +
  d\phi_{2}(\alpha,\beta) \right),
\end{align*}
where the last equality above follows from Lemma \ref{calc3}.
By Proposition \ref{semi-bracket}, the Hamiltonian vector
field of $\brac{\alpha}{\beta}$ is $[v_{\alpha},v_{\beta}]$. Hence we
have:
\[
\tcbrac{\phi_{0}(\alpha)}{\phi_{0}(\beta)}-\phi_{0}(\brac{\alpha}{\beta})=d\phi_{2}(\alpha,\beta).
\]

In degree 1, the bracket $\brac{\cdot}{\cdot}$ is trivial.  Hence it
follows from the definition of $\tcbrac{\cdot}{\cdot}$ and the bilinear
form on $E_{\omega}$ (given in Proposition \ref{2-plectic_courant} ) that
\[
\tcbrac{\phi_{0}(\alpha)}{\phi_{1}(f)}=-\tcbrac{\phi_{1}(f)}{\phi_{0}(\alpha)}=\half
\innerprod{(v_{\alpha},-\alpha)}{(0,-df)} = \phi_{2}(\alpha,df). 
\]
Therefore $\phi_{2}$ is a chain homotopy.

It remains to show the coherence condition (Eq.\ \ref{coherence} in Definition \ref{homo}) is satisfied.
We rewrite the Jacobiator $J_{C}$ as:
\begin{align*}
J_{C}(\phi_0(\alpha),\phi_0(\beta), \phi_0(\gamma))&=-\frac{1}{6}
\innerprod{\cbrac{\phi_{0}(\alpha)}{\phi_0(\beta)}}{\phi_0(\gamma)} +
\cp \\
&=-\frac{1}{6}\innerprod{
  \left([v_{\alpha},v_{\beta}],-\brac{\alpha}{\beta}-dB(\alpha,\beta)
  \right)}{\left(v_{\gamma},-\gamma \right)} + \cp  \\
&=\frac{1}{6} \left( \iota_{[v_{\alpha},v_{\beta}]} \gamma +
  \ip{\gamma}\ip{\beta}\ip{\alpha}\omega + \ip{\gamma}dB(\alpha,\beta)
\right)  + \cp \\
&=-J_{L}(\alpha,\beta,\gamma) + \half
\left(\ip{\gamma}dB(\alpha,\beta) + \cp \right).
\end{align*}
The last equality above follows from Lemma \ref{calc2} and the
definition of the Jacobiator $J_{L}$.
Therefore the left-hand side of Eq.\ \ref{coherence} is
\[
J_{C}(\phi_0(\alpha),\phi_0(\beta), \phi_0(\gamma)) -
\phi_1(J_{L}(\alpha,\beta,\gamma)) =
\half\left(\ip{\gamma}dB(\alpha,\beta) + \cp \right).
\]

By the skew-symmetry of the brackets, the right-hand side of Eq.\
\ref{coherence} can be rewritten as:
\[
\left( \tcbrac{\phi_0(\alpha)}{\phi_2(\beta,\gamma)}+ \cp \right) -
\left(\phi_{2} \left(\brac{\alpha}{\beta},\gamma\right) + \cp \right).
\] 
From the definitions of the bracket, bilinear form and $\phi_2$ we have:
\begin{align*}
\tcbrac{\phi_0(\alpha)}{\phi_2(\beta,\gamma)} + \cp &= \half \innerprod{
  (v_{\alpha}, -\alpha)}{(0,d\phi_{2}(\beta,\gamma)} + \cp\\&= -\half
  \ip{\alpha}dB(\beta,\gamma) + \cp,
\end{align*}
and:
\begin{align*}
\phi_{2} \left(\brac{\alpha}{\beta},\gamma\right) + \cp &=
-\half\left(\iota_{[v_{\alpha},v_{\beta}]} \gamma -
  \ip{\gamma}\ip{\beta}\ip{\alpha}\omega \right)\\
&=-\left(\ip{\alpha}dB(\beta,\gamma) + \cp \right).
\end{align*}
The last equality above follows again from Lemma \ref{calc2}.
Therefore the right-hand side of Eq.\ \ref{coherence} is 
\[
\half\left(\ip{\gamma}dB(\alpha,\beta) + \cp \right).
\]
Hence the maps $\phi_{0}$, $\phi_{1}$, $\phi_{2}$ give a homomorphism of
semistrict Lie 2-algebras.
\end{proof}

We note that Roytenberg \cite{Roytenberg_L2A} has shown that a Courant
algebroid defined using Definition \ref{alt_courant_algebroid} with the
bilinear operation $\circ$ induces a hemistrict Lie 2-algebra on the
complex $C$ described in Theorem \ref{courant_L2A} above. A hemistrict
Lie 2-algebra is a Lie 2-algebra in which the skew-symmetry holds
up to isomorphism, while the Jacobi identity holds as an equality.
We have proven in previous work \cite{BHR} that
a 2-plectic structure also gives rise to a hemistrict Lie 2-algebra on
the complex described in Theorem \ref{semistrict}. One can show that all
results presented above, in particular Theorem \ref{main_thm}, carry over
to the hemistrict case.

\section{Hamiltonian 1-forms as infinitesimal symmetries of the Courant algebroid}
Given a 2-plectic manifold $(M,\omega)$, the Lie 2-algebra of
observables $L(M,\omega)$ identifies particular infinitesimal symmetries
of the corresponding Courant algebroid $E_{\omega}$ via the embedding
described in the proof of Theorem \ref{main_thm}. To see this, we
first recall some basic facts concerning automorphisms of exact
Courant algebroids. The presentation here follows the work of
Bursztyn, Cavalcanti, and Gualtieri
\cite{Bursztyn-Cavalcanti-Gualtieri}.

\begin{definition} \label{automorphism}
  Let $E \rightarrow M$ be a Courant algebroid with bilinear form
  $\innerprod{\cdot}{\cdot}$, bracket $\cbrac{\cdot}{\cdot}$, and
  anchor $\rho \maps E \to TM$. An {\bf automorphism} is a bundle
  isomorphism $F\maps E \to E$ covering a diffeomorphism $\varphi \maps
  M \to M$ such that
\begin{enumerate}
\item{$\varphi^{\ast} \innerprod{F(e_1)}{F(e_{2})} =
    \innerprod{e_{1}}{e_{2}}$,
}
\item{$F \left(\cbrac{e_{1}}{e_{2}}\right)=\cbrac{F(e_{1})}{F(e_{2})}$,
}
\item{$\rho\left(F(e_{1})\right)=\varphi_{\ast}\left(\rho(e_{1})
    \right)$.
}
\end{enumerate}
\end{definition}

Consider the exact Courant algebroid $E_{\omega}$ described in Section
\ref{twist} with underlying vector bundle $TM \oplus T^{\ast}M
\rightarrow M$ and \v{S}evera class $[\omega] \in
H^{3}_{\mathrm{DR}}(M)$. Given a 2-form $B \in \Omega^{2}(M)$, one can
define a bundle isomorphism 
\[
\exp B \maps TM \oplus T^{\ast}M
\to TM \oplus T^{\ast}M 
\]
by
\[
\exp B\left(v,\alpha \right)= \left(v,\alpha + \iota_{v}B\right).
\]
The map $\exp B$ is known as a `gauge transformation'. It covers the
identity $\id \maps M \to M$ and therefore is compatible (in the sense
of Definition \ref{automorphism}) with the anchor
$\rho(v,\alpha)=v$. Since $B$ is skew-symmetric, $\exp B$ preserves
the bilinear form $\innerprod{(v_{1},\alpha_{1})}{(v_{2},\alpha_{2})}=
\ip{1}\alpha_{2} + \ip{2}\alpha_{1}$. However a simple computation
shows that $\exp B$ preserves the bracket
$\cbrac{\cdot}{\cdot}_{\omega} $ (defined in Eq.\
\ref{twisted_bracket}) if and only if $B$ is a closed 2-form:
\[
\cbrac{\exp B \left(v_{1},\alpha_{1}\right)}{\exp
  B \left(v_{2},\alpha_{2}\right)}_{\omega} = \exp B \left
  (\cbrac{(v_{1},\alpha_{1})}{(v_{2},\alpha_{2})}_{\omega+dB} \right).
\]

Given a diffeomorphism $\varphi \maps M \to M$ of the base space, one can
define a bundle isomorphism $\Phi \maps TM \oplus T^{\ast}M
\to TM \oplus T^{\ast}M$ by
\[
\Phi\left( v,\alpha \right) =\left(\varphi_{\ast} v,
  \left(\varphi^{\ast}\right)^{-1} \alpha \right).
\]
The map $\Phi$ satisfies conditions 1 and 3 of Definition
\ref{automorphism} but does not preserve the bracket in general:
\[
\cbrac{\Phi \left(v_{1},\alpha_{1}\right)}{\Phi \left(v_{2},\alpha_{2}
  \right)}_{\omega} = \Phi\left
  (\cbrac{(v_{1},\alpha_{1})}{(v_{2},\alpha_{2})}_{\varphi^{\ast} \omega} \right).
\]
Bursztyn, Cavalcanti, and Gualtieri
\cite{Bursztyn-Cavalcanti-Gualtieri} showed that any automorphism $F$
of the exact Courant algebroid $E_{\omega}$ must be of the form
\begin{equation}
F=\Phi \exp B, \label{exact_auto1}
\end{equation}
where $\Phi$ is constructed from a diffeomorphism $\varphi \maps M \to M$
such that 
\begin{equation}
\omega - \varphi^{\ast}\omega = dB.  \label{exact_auto2}
\end{equation}
This classification of automorphisms allows one to classify the
infinitesimal symmetries as well. Let 
\[
F_{t}=\Phi_{t}\exp tB = \left(\varphi_{t\ast} \exp tB,
  \left(\varphi_{t}^{\ast}\right)^{-1} \exp tB\right)
\]
be a 1-parameter family of automorphisms of the Courant algebroid
$E_{\omega}$ with $F_{0}=\id_{E_{\omega}}$. Let $u \in \mathrm{Vect}(M)$ be
the vector field that generates the flow $\varphi_{-t}$. Then differentiation gives:
\[
\left. \frac{dF_{t}}{dt} \left( v,\alpha \right) \right \vert_{t=0} =
\left ([u,v], \L_{u}\alpha + \iota_{v}B \right).
\]
Since $\omega - \varphi_{t}^{\ast}\omega = tdB$, it
follows that $u$ and $B$ must satisfy the equality:
\begin{equation}
\L_{u}\omega = dB. \label{derivation}
\end{equation}
These infinitesimal transformations are called \textbf{derivations}
\cite{Bursztyn-Cavalcanti-Gualtieri} of the Courant algebroid
$E_{\omega}$, since they correspond to linear first order differential
operators which act as derivations of the non-skew-symmetric bracket:
\begin{align}
(v_{1},\alpha_{1}) \circ_{\omega} (v_{2},\alpha_{2}) &=
\cbrac{(v_{1},\alpha_{1})}{(v_{2},\alpha_{2})}_{\omega} + \half
d\innerprod{(v_{1},\alpha_{1})}{(v_{2},\alpha_{2})}. \label{alt_twist}\\
&=\left([v_{1},v_{2}], \L_{v_{1}}\alpha_{2} - \ip{2}d\alpha_{1} +
  \ip{2} \ip{1}\omega \right).
\end{align}
In general, derivations are
pairs $(u,B) \in \mathrm{Vect}(M) \oplus \Omega^{2}(M)$ satisfying Eq.\
\ref{derivation}. They act on global sections $(v,\alpha) \in
\Gamma(E_{\omega})$ by:
\[
(u,B) \cdot (v,\alpha) = \left ([u,v], \L_{u}\alpha + \iota_{v}B
\right).
\]

Global sections themselves naturally act as derivations via an \textbf{adjoint action} \cite{Bursztyn-Cavalcanti-Gualtieri}.
Given $\left(u,\beta \right) \in \Gamma (E_{\omega})$ let $B$
be the 2-form
\begin{equation}
B = -d\beta + \iota_{u} \omega. \label{B}
\end{equation}
Define $\ad_{(u,\beta)} \maps \Gamma (E_{\omega}) \to \Gamma
(E_{\omega})$ by
\begin{equation}
\ad_{(u,\beta)} (v,\alpha)= (u,B) \cdot (v,\alpha)=\left ([u,v],
  \L_{u}\alpha + \iota_{v}\left(-d\beta + \iota_{u} \omega \right)
\right).
\end{equation}
One can see this is indeed the adjoint action in the usual sense if one considers
the non-skew-symmetric bracket given in Eq.\ \ref{alt_twist}:
\[
\ad_{(u,\beta)} (v,\alpha) = (u,\beta) \circ_{\omega} (v,\alpha).
\]

Recall that in the proof of Theorem \ref{main_thm} we constructed a
homomorphism of Lie 2-algebras using the map
$\phi_{0} \maps \ham \to \Gamma(E_{\omega})$  defined by 
\[
\phi_{0}(\alpha)=(v_{\alpha},-\alpha),
\]
where $v_{\alpha}$ is the Hamiltonian vector field corresponding to $\alpha$. 
Comparing Eq.\ \ref{B} to Definition \ref{hamiltonian} of a Hamiltonian 1-form, we
see that a section $(u,\beta) \in \Gamma(E_{\omega})$ is in the image
of the map $\phi_{0}$ if and only if its adjoint action
$\ad_{(u,\beta)}$ corresponds to the pair $(u,0)  \in \mathrm{Vect}(M)
\oplus \Omega^{2}(M)$. This implies that $\ad_{(u,\beta)}$ preserves
the 2-plectic structure on $M$ and that $-\beta$ is a Hamiltonian 1-form
with Hamiltonian vector field $u$. Also if $u$ is
complete, then Eqs.\ \ref{exact_auto1} and \ref{exact_auto2} imply
that the 1-parameter family $F_{t}$ of Courant algebroid automorphisms
generated by $\ad_{(u,\beta)}$ corresponds to a 1-parameter family of diffeomorphisms
$\varphi_{t} \maps M \to M$ which preserve the
2-plectic structure:
\[
\varphi^{\ast}_{t} \omega = \omega.
\]  
In analogy with symplectic geometry, we call such automorphisms
\textbf{Hamiltonian 2-plectomorphisms}.
 
We provide the following proposition as a summary of the discussion presented in this
section:
\begin{prop}
Let $(M,\omega)$ be a 2-plectic manifold and let $E_{\omega}$
be its corresponding Courant algebroid. There is a one-to-one
correspondence between the Hamiltonian 1-forms $\ham$ on
$(M,\omega)$ and sections $(u,\beta)$ of $E_{\omega}$ whose
adjoint action satisfies the equality 
\[
\ad_{(u,\beta)}(v,\alpha)=\left ( \L_{u} v, \L_u{\alpha} \right).
\]
\end{prop}

\section{Conclusions}
We suspect that the results presented here are preliminary and
indicate a deeper relationship between 2-plectic geometry and the theory of
Courant algebroids. For example, the discussion of connections and
curvature in Section \ref{class} is reminiscent of the theory of gerbes
with connection \cite{Brylinski}, whose relationship with Courant
algebroids has been already studied \cite{Bressler-Chervov,Severa1}. In
2-plectic geometry, gerbes have been conjectured to play a role in the
geometric quantization of a 2-plectic manifold \cite{BHR}. It will be
interesting to see how these different points of view complement each
other.

In general, much work has been done on studying the geometric
structures induced by Courant algebroids (e.g.\ Dirac structures,
twisted Dirac structures). Perhaps this work can aid 2-plectic
geometry since many geometric structures in this context are somewhat
less understood or remain ill-defined (e.g.\ the notion of a
2-Lagrangian submanifold or 2-polarization).

On the other hand, $n$-plectic manifolds are well
understood in the role they play in classical field theory
\cite{GIMM}, and are also understood algebraically in the sense that
an $n$-plectic structure gives an $n$-term $L_{\infty}$-algebra on a
chain complex of differential forms \cite{CR}. Perhaps these insights
can aid in understanding `higher' analogs of Courant algebroids
(e.g.\ Lie $n$-algebroids) and complement similar ideas discussed by
\v{S}evera in \cite{Severa2}.
\section{Acknowledgments}
We thank John Baez, Yael Fregier, Dmitry Roytenberg, Urs Schrieber,
James Stasheff and Marco Zambon for helpful comments, questions, and conversations.

\end{document}